%% file: main.tex
\documentclass[conference]{IEEEtran}
\IEEEoverridecommandlockouts
\usepackage[T1]{fontenc}
\usepackage{subcaption}
\usepackage{amsmath}
\usepackage{amsfonts}
\usepackage{amsthm}
\usepackage{physics}
\usepackage{amsmath}
\usepackage{tikz}
\usepackage{mathdots}
\usepackage{yhmath}
\usepackage{cancel}
\usepackage{color}
\usepackage{siunitx}
\usepackage{array}
\usepackage{multirow}
\usepackage{amssymb}
\usepackage{gensymb}
\usepackage{tabularx}
\usepackage{extarrows}
\usepackage{booktabs}

\usepackage{amssymb}
\usepackage{latexsym}   
\usepackage{graphicx}
\usepackage{color}
\usepackage{hyperref}
\usepackage[linesnumbered,ruled,vlined]{algorithm2e}
\usepackage{multirow}

\usepackage{makecell}

\newcommand{\rmat}{\textsc{R-MAT}\xspace}

\newcommand{\TW}{\text{TW}\xspace}
\newcommand{\tec}{\text{Tectonic}\xspace}

\newcommand{\hide}[1]{}

\usepackage{colortbl}

\def\BibTeX{{\rm B\kern-.05em{\sc i\kern-.025em b}\kern-.08em
    T\kern-.1667em\lower.7ex\hbox{E}\kern-.125emX}}

\newtheorem{theorem}     {Theorem}

\newtheorem{proposition} {Proposition}

\begin{document}
\title{Parallel Motif-Based Community Detection} 


\author{\IEEEauthorblockN{Tianyi Chen}
\IEEEauthorblockA{Celonis Inc. \\
t.chen@celonis.com}
\and
\IEEEauthorblockN{Charalampos E. Tsourakakis}
\IEEEauthorblockA{relational\underline{AI} \\
charalampos.tsourakakis@relational.ai}
}

\maketitle
\begin{abstract}
\input{abstract}
\end{abstract}

\section{Introduction}
\input{intro}

\section{Related Work}
\label{sec:related}
\input{related}

\section{Proposed Algorithms}
\label{sec:proposed}
\input{proposed}

\section{Experiments}
\label{sec:exp}
\input{exp}

\section{Conclusion}
\input{conclusion}

\bibliographystyle{splncs04}
\bibliography{ref}


\end{document}

%% file: abstract.tex
Community detection is a central task in graph analytics. Given the substantial growth in graph size, scalability in community detection continues to be an unresolved challenge~\cite{shi2021scalable}. Recently, alongside established methods like Louvain~\cite{Blondel2008louvain} and Infomap~\cite{rosvall2008maps}, motif-based community detection has emerged~\cite{satuluri2011local,benson2016higher,Tsourakakis2017scalable}. Techniques like Tectonic~\cite{Tsourakakis2017scalable} are notable for their advanced ability to identify communities by pruning edges based on motif similarity scores and analyzing the resulting connected components.

In this study, we perform a comprehensive evaluation of community detection methods, focusing on both the quality of their output and their scalability. Specifically, we contribute an open-source parallel framework for motif-based community detection based on a shared memory architecture. We conduct a thorough comparative analysis of community detection techniques from various families among state-of-the-art methods,  including Tectonic~\cite{Tsourakakis2017scalable}, label propagation~\cite{usha2007near}, spectral clustering, Louvain~\cite{Blondel2008louvain},  LambdaCC~\cite{veldt2018correlation}, and Infomap~\cite{rosvall2008maps} on graphs with up to billions of edges.  A key finding of our analysis is that motif-based graph clustering provides a good balance between performance and efficiency. Our work provides several novel insights. Interestingly, we pinpoint biases in prior works in evaluating community detection methods using the top 5K groundtruth communities from SNAP only, as these are frequently near-cliques.  Our empirical studies lead to rules of thumb threshold picking strategies that can be critical for real applications.  Finally, we show that Tectonic can fail to recover two well-separated clusters. To address this, we suggest a new similarity measure based on counts of triangles and wedges (\TW) that prevents the over-segmentation of communities by Tectonic.

%% file: intro.tex
Graph communities are groups of nodes that are densely connected inside and sparsely connected outside. Such structures typically cluster nodes with similar properties and have been used to detect protein complexes~\cite{gavin2006proteome}, make content recommendations~\cite{feng2015personalized}, and entity resolution~\cite{chierichetti2014correlation}.  Community detection is routinely performed in graphs  and forms a component of various commercial knowledge graph databases such as RelationalAI~\cite{rel} and Neo4j~\cite{neo4j} and has received a significant amount of attention from the research community~\cite{fortunato2010community}.  Community detection is evaluated according to two main criteria. The first one is the output quality, which is assessed either through groundtruth knowledge or through various measures such as modularity~\cite{Girvan2002community}, conductance~\cite{leskovec2008statistical} among several others~\cite{liu2019evaluation,fortunato2010community,orecchia2022practical}. Ground-truth knowledge is typically available on synthetic datasets and a limited number of real-world networks.     The second criterion is scalability. Due to the growing size of networks, scalability continues to be a significant challenge.  In recent years, a framework has emerged that, despite its simplicity, provides a good balance between scalability and output quality.  This framework was inaugurated by Satuluri et al.~\cite{satuluri2011local}. Their idea is to reweigh  each edge $e=(u,v)$ in the graph by the Jaccard similarity of the neighborhoods $N(u), N(v)$ respectively, i.e., $|N(u)\cap N(v)|\over |N(u)\cup N(v)|$ and then remove the low-weight edges before applying a community detection algorithm 
on the sparsified graph to achieve speed up.    
Tsourakakis et al.~\cite{Tsourakakis2017scalable}  proposed Tectonic, a similar heuristic that measures the similarity of two connected nodes $(u,v)$ by $t(u,v)\over \deg(u)+\deg(v)$, where $t(u,v)$ is the number of triangles edge $(u,v)$ participates in and $\deg(u)$ is the degree of node $u$.   After removing low-similarity edges, they output the connected components as clusters. Their work introduced a framework for evaluating community detection methods using precision and recall and showed impressive accuracies on graphs with known groundtruth communities, see also Shi et al.~\cite{shi2021scalable}.  It is worth noting that the Tectonic heuristic performs significantly better than other motif-based spectral clustering, which reweighs each edge by the number of participating triangles before applying standard spectral clustering~\cite{benson2016higher,Tsourakakis2017scalable}.

In this work, we provide an extensive, in-depth evaluation of scalable community detection methods. Our key finding is that motif-based community detection methods provide a good trade-off between output quality and scalability. We provide an open-source implementation of our method and other competitors in MPI. Additionally, we address a commonly overlooked issue in motif-based community detection methods: selecting the threshold parameter. We provide a robust rule of thumb for determining this parameter. We also pinpoint a community oversegmentation issue inherent to Tectonic~\cite{Tsourakakis2017scalable} and propose a new, intuitive heuristic \TW that performs well empirically. This heuristic relies solely on counting triangles and wedges, and we provide theoretical evidence to support its effectiveness.

We conduct an extensive evaluation of motif-based community detection methods including \TW and other major competitors on various real-world graphs with or without groundtruth. In addition, we also use synthetic \rmat~\cite{Chakrabarti2004RMATAR} graphs to evaluate scalability. As a secondary contribution, we uncover a hidden bias in the evaluation of several prior works~\cite{yang2012define,Tsourakakis2017scalable,shi2021scalable} that focus on the top 5k ground-truth communities from SNAP~\cite{snapnets}. These top communities tend to be cliques or near-cliques (see Figure~\ref{fig:dblp_com_density_compare}), which contrasts with the majority of other ground-truth communities. In a nutshell, we find that motif-based community detection methods, including our \TW-based framework, are a viable alternative to popular choices such as Louvain's method~\cite{Blondel2008louvain} or InfoMap~\cite{rosvall2008maps}.

\textbf{Notations.} We focus  on undirected, unweighted graphs $G=(V,E)$ with $n=|V|$ nodes and $m=|E|$ edges. For each node $u\in V$, we denote its neighbors as $N(u)$ and its degree as $\deg(u)=|N(u)|$. For any edge $e=(u,v)\in E$, the number of triangles and wedges it participates in are denoted as $t(u,v)=|N(u)\cap N(v)|$ and $wedge(u,v)=|N(u)\cup N(v)|-|N(u)\cap N(v)|-2$, respectively. We use $c_u$ to represent the community to which node $u$ belongs. We denote the arboricity of a graph $G$ as $\alpha (G)$, which is used to bound the complexity of clique listing tasks~\cite{chiba1985arboricity}.

%% file: related.tex

\subsection{Community detection} Community detection, also known as graph clustering, is an extensively studied data mining task~\cite{fortunato2010community}. Various formulations and evaluation criteria have been proposed~\cite{fortunato2010community}. One of the most popular formulations is modularity optimization~\cite{Girvan2002community}. The modularity of a partition is defined as the difference between the number of edges within groups and the expected number of edges within groups under a null model.  The two foundational null models are the $G(n,p)$  random graph model and the configuration model~\cite{frieze2023random}, with the latter being more prevalent in practice due to the skewed degree distributions of real-world networks. Specifically, by assuming the configuration model, the probability of an edge $(i,j)$ is equal to $\frac{\deg(i) \deg(j)}{2m}$. Given the adjacency matrix $A$ of the graph, this yields that the modularity $
Q=\frac{1}{2m} \sum_{u,v} \big( A_{uv}-\frac{\deg(u) \deg(v)}{2m} \big)\delta (c_u,c_v)$, where
$\delta(c_u,c_v)$ equals 1 if $c_u=c_v$ and 0 otherwise.
While maximizing the modularity is NP-hard~\cite{brandes2007modularity}, many heuristics have been proposed, including Louvain~\cite{Blondel2008louvain}, one of the most successful community detection methods. Given its efficiency and accuracy~\cite{hric2014community}, practitioners widely use Louvain on various knowledge graph platforms such as Neo4j~\cite{neo4j_louvain} to solve real-world community detection challenges. Modularity optimization is also widely studied under the parallel setting, with either shared-memory~\cite{staudt2015engineering,halappanavar2017scalable,shi2021scalable} or distributed memory~\cite{zeng2015parallel,ghosh2018distributed}.

Veldt et al.~\cite{veldt2018correlation} recently proposed a community detection framework called LambdaCC by converting an unsigned graph into a signed graph and finding a community assignment that minimizes a modified correlation clustering objective~\cite{bansal2004correlation}. Their framework has a resolution parameter $\lambda$ that, for different values, gives important objectives such as the sparsest cut~\cite{arora2010logn} and modularity optimization as special cases.  Their framework is seminal and has provided not only theoretical insights but also state-of-the-art interpolation among the various well-known community detection objectives. Shi et al.~\cite{shi2021scalable} implemented the LambdaCC framework by designing a non-trivial parallel algorithm that resembles Louvain-like optimization for modularity.   This work experiments  with several heuristics to optimize the performance, including synchronous label updates, reducing the set of vertices to consider moving, and multilevel refinement. Their experimental results exhibited that their parallel framework has great parallel scalability and outcompetes other baselines in terms of output quality according to the ground truth communities. Other clustering objectives have been proposed and extensively studied. For instance, conductance measures the number of edges being cut by the clusters, normalized by the volume of clusters. While it is NP-hard to minimize graph conductance, Cheeger's inequality~\cite{alon1985lambda1} 
provides a polynomial time approximation algorithm ~\cite{ng2001spectral}. 

The map equation introduced by Rosvall et al.~\cite{rosvall2008maps}  encodes a random walker’s trajectory with codewords assigned to nodes, aiming to minimize the average per-step description length. It can be greedily optimized by InfoMap~\cite{mapequation2023software}, a refined framework of Louvain. Many community detection heuristics are inspired by the properties of the community structures rather than optimizing some objectives. For instance, label propagation was adapted as a heuristic initially by~\cite{usha2007near}. The authors described a simple procedure where nodes adopt the most frequent community assignment among their neighbors until stable communities emerge. 

In recent years, motif-based community detection methods have emerged. Notably, Benson et al. and Tsourakakis et al. introduced motif-based spectral clustering by defining the concept of motif conductance~\cite{benson2016higher,Tsourakakis2017scalable}.   In terms of well-performing methods in practice, Tectonic~\cite{Tsourakakis2017scalable} stands out and has been used as a key competitor in recent scalability studies~\cite{shi2021scalable}. It assigns to each edge a similarity score equal to $t(u,v)\over \deg(u)+\deg(v)$, then removes the low-similarity edges and outputs the connected components as clusters. A key advantage of this framework is that it relies on local graph characteristics that can be efficiently computed for massive graphs in different parallel or distributed architectures.  However, the authors of the original paper studied Tectonic only on graphs with groundtruth, where they could observe the precision and recall as the threshold changed. The choice of the threshold parameter is a challenge that has not been addressed in prior works, and we discuss it in this work.

\subsection{Parallel computation: model and primitives} The dynamic multithreading model represents a parallel computation as a directed acyclic graph (DAG), where vertices represent instructions and edges represent dependencies. The runtime scheduler can execute a computation only when its preceding nodes in the DAG are finished. We use the work-depth model in our algorithm analysis. The  {\em work} $W$ of computation is the number of vertices, and the  {\em depth} $D$ is the length of the longest path in the DAG. If $P$ processors are available, one can bound the running time by $O(W/P+D)$ using Brent's scheduling theorem~\cite{blelloch1996parallel}.
We introduce two parallel primitives used in this work~\cite{blelloch1996parallel}. 
\textit{Merge} takes two sorted sequences $A$ and $B$ of length $n$ and $m$, and returns a sorted sequence containing the union of the elements in $A$ and $B$. It can be implemented in $O(n+m)$ work and $O(\log (m+n))$ depth. It can also be modified to return the intersection of the elements of two sorted sequences in the same complexity. Sorting a sequence takes $O(n\log n)$ work and $O(\log ^{3/2}n)$ depth with high probability.


\subsection{Parallel algorithms} Shun and Tangwongsan~\cite{shun2015multicore} provided a parallel implementation of parameter-free triangle counting algorithm with $O(m^{3/2})$ work and $O(\log^{3/2}m)$ depth. The algorithm first ranks nodes based on their degrees, then transforms the graph into directed based on the rank and counts the triangles by intersecting adjacency lists. Algorithms for other architectures such as MapReduce are also available~\cite{biswas2020massively}.  Shun et al.~\cite{shun2014simple} proposed a parallel algorithm for finding connected components that take linear work and polylogarithmic depth. It repetitively applies a parallel low-diameter graph decomposition~\cite{miller2013parallel}  based on BFS and contracts the graph until no edges are left.

%% file: proposed.tex
\subsection{Motif-based Community Detection}


The motif-based community detection framework is summarized in Algorithm~\ref{alg:framework}. Initially, the framework sparsifies the input graph $G$ using a user-specified edge similarity function $sim$ and a threshold value $\delta$. Subsequently, it identifies connected components which constitute the final output. Table~\ref{tab:simfunctions} presents a selection of similarity functions we use and their computational complexity. Despite the fact that all measures are already well known in the literature except for \TW, we provide a detailed evaluation of them within the context of Algorithm~\ref{alg:framework}. We discuss these measures in the following. 

 \begin{algorithm}[]
\caption{Motif-based Community detection}\label{alg:framework}
\SetKwInOut{Input}{Input}
\SetKwInOut{Output}{Output}
\Input{\ $G=(V,E)$, threshold $\delta \in \mathbb{R}^+$, similarity function $sim : E \to \mathbb{R}^+$}
\Output{\ Graph partitions}
For each edge $e\in E$, if $\mathrm{sim}(e)<\delta$\ remove it from $G$\;
CCs $\leftarrow$ connected components of sparsified $G$\; 
\Return CCs
\end{algorithm}

\begin{table}[]
    \centering
    \footnotesize
    \begin{tabular}{|c|c|c|} \hline
      $sim$ function & Definition &  Complexity\\ \midrule 
       \begin{tabular}{@{}c@{}}Effective resistance \\ (Eff. Res.)~\cite{broder1989generating}\end{tabular}  & $-R_{eff}(u,v)$ & $O(m \sqrt{\log n})$  \\ 
       \begin{tabular}{@{}c@{}}Betweenness  \\ centrality (BC)\end{tabular}  & $-BC(u,v)$   & APSP~\cite{abboud2014subcubic}\\ 
        \midrule 
       Tectonic~\cite{Tsourakakis2017scalable}  & $\frac{t(u,v)}{deg(u)+deg(v)}$ &$O(m \alpha(G))$ \\
       Jaccard coeff.~\cite{satuluri2011local} &  $\frac{t(u,v)}{deg(u)+deg(v)-t(u,v)}$ & $O(m \alpha(G))$ \\  
       \begin{tabular}{@{}c@{}}\#  triangles ($K_3$) \\ \cite{sotiropoulos2021triangle,Tsourakakis2017scalable,benson2016higher}\end{tabular}  &  $t(u,v)$ & $O(m \alpha(G))$  \\ 
       \begin{tabular}{@{}c@{}}$\kappa$-clique participation \\ ~\cite{Tsourakakis2017scalable,benson2016higher}\end{tabular}  & $K_\kappa(u,v)$ & $O(m \alpha(G)^{\kappa-2})$ \\ \hline 
       {\bf TW} (this paper) & $t(u,v)-wedge(u,v)$ & $O(m \alpha(G))$ \\     \hline

    \end{tabular}
    \caption{Measures  we use as the similarity function $sim$ in Algorithm~\ref{alg:framework}. }
    \label{tab:simfunctions}
\end{table} 

We include two well-known similarity measures that are not a function of motif counts but are closely related to motifs. The effective resistance of an edge equals the probability of this edge being included in a random spanning tree of the graph~\cite{broder1989generating}. It is upper bounded by $\frac{2}{2+t(u,v)}$ as shown in~\cite{sotiropoulos2021triangle}. An edge within a community intuitively has lower effective resistance than an edge across two communities~\cite{alev2017graph}.  On the contrary, if the effective resistance is closer to one, the edge is more likely to connect two communities. It is worth pointing out that this intuition holds frequently but not always, see~\cite{sotiropoulos2021triangle}.   Edge betweenness is a measure used in network analysis to quantify the importance of an edge in a graph. It is defined as the fraction of all shortest paths that pass through this edge. Intuitively, edges across two communities are more likely to be used in the shortest paths, as they act as bridges between communities. Negations are taken on both functions in Table~\ref{tab:simfunctions} so they fit in Algorithm~\ref{alg:framework}.


Tectonic achieves competitive, state-of-the-art performance on community detection~\cite{Tsourakakis2017scalable}. Recall that the Jaccard similarity of an edge $(u,v)$ is simply the Jaccard coefficient of the neighborhoods $N(u), N(v)$, i.e., 

$$ J(e) = J(N(u), N(v)) = \frac{t(u,v)}{\deg(u)+\deg(v)-t(u,v)}. $$

We note that thresholding Jaccard similarity~\cite{satuluri2011local}  is equivalent to Tectonic up to a re-parameterization. Therefore, we do not include it as we use a wide range of thresholds for Tectonic.  We state this as the next  Proposition~\ref{prop:equiv}.  

\begin{proposition}\label{prop:equiv}
     Tectonic is a reparameterization of Jaccard edge similarity.
\end{proposition}

\begin{proof}
    For any threshold value $\delta$ for Jaccard edge similarity, an edge $(u,v)$ is removed if  the following holds: 
    \begin{align}
      & \frac{t(u,v)}{\deg(u)+\deg(v)-t(u,v)}<\delta && \nonumber\\
      \rightarrow & (1+\delta)t(u,v)<\delta (\deg(u)+\deg(v)) && \nonumber \\
       \rightarrow & \frac{t(u,v)}{\deg(u)+\deg(v)}<\frac{\delta}{1+\delta} && \nonumber
    \end{align}
    
Observe that if we interpret $\delta$ as a probability, then the range of $\frac{\delta}{1 + \delta}$ is $[0,0.5]$. The equivalence is immediately derived from the bijectiveness of the function $f(\delta)=\frac{\delta}{1+\delta}$ in the interval $\delta \in [0,1]$.  

\end{proof}

Sotiropoulos and Tsourakakis~\cite{sotiropoulos2021triangle} used the upper bound $\frac{2}{2+t(u,v)}$ as a proxy for the effective resistance of an edge.  The number of triangles $t(u,v)$ were used for reweighing edges before applying spectral clustering as an efficient way to minimize the triangle conductance of a graph~\cite{benson2016higher,Tsourakakis2017scalable}.  
Another measure we experiment with is $\kappa$-clique participation~\cite{Tsourakakis2017scalable}, a generalization of triangle participation. Due to the computation complexity, in practice, we only consider $K_4$, i.e., a complete subgraph of size $4$.

\begin{figure}
    \centering
    \begin{subfigure}[]{0.25\linewidth}
        \centering
        \includegraphics[width=\linewidth]{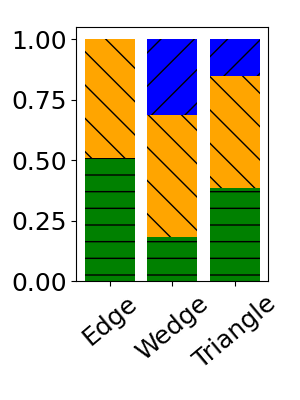}
        \caption{}
    \end{subfigure} \hfill 
    \begin{subfigure}[]{0.25\linewidth}
        \centering
        \includegraphics[width=\linewidth]{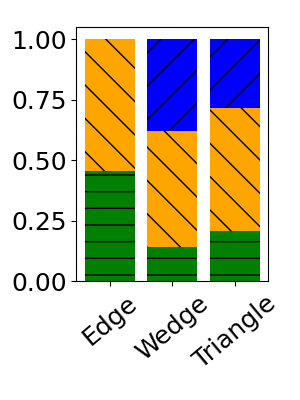}
        \caption{}
    \end{subfigure}  \hfill
    \begin{subfigure}[]{0.25\linewidth}
        \centering
        \includegraphics[width=\linewidth]{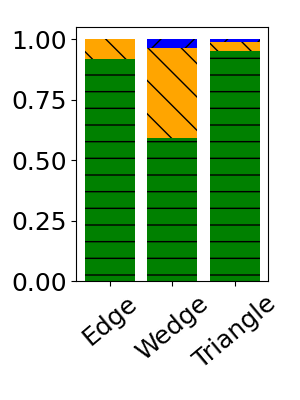}
        \caption{}
    \end{subfigure} 
    \begin{subfigure}[]{0.4\linewidth}
        \centering 
        \includegraphics[width=\linewidth]{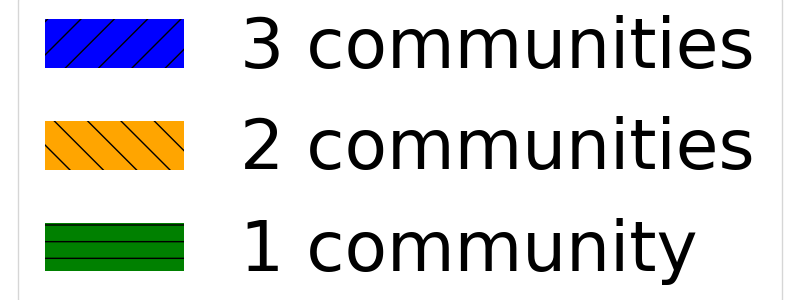}
    \end{subfigure} \hfill
     \begin{subfigure}[]{0.4\linewidth}
        \centering
        \includegraphics[width=\linewidth]{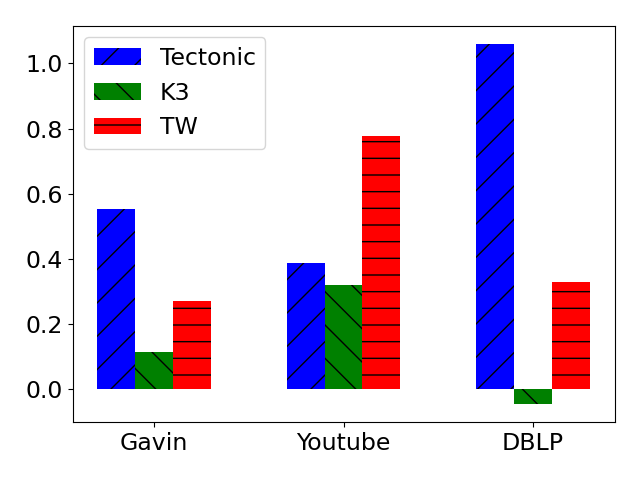}
        \caption{}
    \end{subfigure} 
    \caption{Fractions of motifs cut by groundtruth communities in (a) Gavin, (b) Youtube, and (c) DBLP. (d) The average difference between similarity scores of edges inside and across communities, normalized by the standard deviation. }
    \label{fig:motivation}
\end{figure}

Most efficient similarity measures are calculated based on the number of triangles participated. The intuition comes from the fact that the triangle is a significant network motif~\cite{chen2022algorithmic} in the communities, according to empirical studies~\cite{benson2016higher}. However, little research has looked into other small motifs, especially wedges, with few exceptions that focus on computational aspects of triangle counting~\cite{seshadhri2013triadic}. Figures~\ref{fig:motivation}(a)-(c) show the distributions of motifs of size at most three in three datasets with ground-truth communities, see Section~\ref{sec:exp} for details of the datasets.
While the proportions of motifs being cut vary greatly, we consistently observe that the fraction of wedges being cut is more than that of triangles with a non-trivial difference. This implies that edges participating in more wedges will likely go across communities and thus should receive lower similarity scores. Our new similarity function, called \TW, achieves this. Specifically, it assigns a score  to each edge $e=(u,v)$

\begin{equation}
    \label{eq:tw}
TW(e)=t(e)- wedges(e).    
\end{equation}


Figure~\ref{fig:motivation}(d) demonstrates the average difference inside and across communities between similarity scores, including Tectonic, the number of triangles, and \TW, respectively. We observe both \TW and \tec give strong signals on whether the edges are inside communities or across. Specifically, \TW outperforms Tectonic on Youtube, a dataset where triangles are located more across communities according to Figure~\ref{fig:motivation}(b). On the contrary, $K_3$s provide a weaker or even wrong signal in datasets like DBLP.



\subsection{Triangle-Wedges (\TW) Method: Theoretical Insights}\label{subsec:theory}

As mentioned earlier, \tec has performed remarkably on community detection benchmarks. Nonetheless, we show that it can fail even in a simple setting.  In particular, consider an instance of a stochastic blockmodel (SBM)~\cite{abbe2015exact} as follows:
there exist two community blocks $B_1, B_2$ each with $n$ nodes, i.e., $|B_1|=|B_2|=n$. Each pair of nodes $u,v$ within a block $B_i$ are connected with probability $p_i, i=1,2$. Any pair of nodes across the two blocks is connected with probability $q$. We assume that   $p_1,p_2>q$, but $p_1\neq p_2$ as in the standard SBM. 
The following theorem shows a range of parameters for which \tec fails, but \TW succeeds in recovering the groundtruth communities. It is worth emphasizing that this does not mean that \TW is always better than \tec but is a viable alternative with an interpretable intuition.  

\begin{theorem}
\label{thm:sbm}
Let $p_1=\frac{2\log n}{n}$ to ensure connectivity with high probability. Furthermore, let $q=p_1/2$ and $p_1 \ll p_2$. Then \TW recovers the two blocks, while Tectonic fails in expectation respectively. 
\end{theorem}



\begin{table*}[]
    \centering
    \caption{Expectation of various measures for each edge in an inhomogeneous SBM edge depending on the communities of its endpoints, normalized by $n$. We further assume $n-2 \approx n$ since we care about asymptotics. }
    \label{tab:expected}
    \begin{tabular}{c|c|c|c|c|c}
         & \# triangles & \# wedges & $d_u+d_v$ & $\TW(u,v)$ & $Tec_{\delta}(u,v)$ \\ \hline
        \makecell{Inside\\ $B_1$} & $p_1^2+q^2$ & $2(p_1(1-p_1)+q(1-q))$ & $2(p_1+q) $ & \makecell{$3p_1^2-2p_1$\\ $+3q^2-2q$} & \makecell{$p_1^2+q^2$\\ $-2\delta (p_1+q)$} \\ \hline
        \makecell{Across\\ $B_1$, $B_2$} & $p_1q+p_2q$ & \makecell{$p_1(1-q)+p_2(1-q)+$ \\ $q(1-p_1)+q(1-p_2) $} & $p_1+p_2+2q$ & \makecell{$3p_1q+3p_2q$\\ $-p_1-p_2-2q$}  & \makecell{$p_1q +p_2q$\\ $-\delta (p_1+p_2+2q)$} \\
    \end{tabular}
\end{table*}

\begin{figure*}[htbp]
    \centering
    \begin{subfigure}[]{0.22\linewidth}
        \centering
        \includegraphics[width=\linewidth]{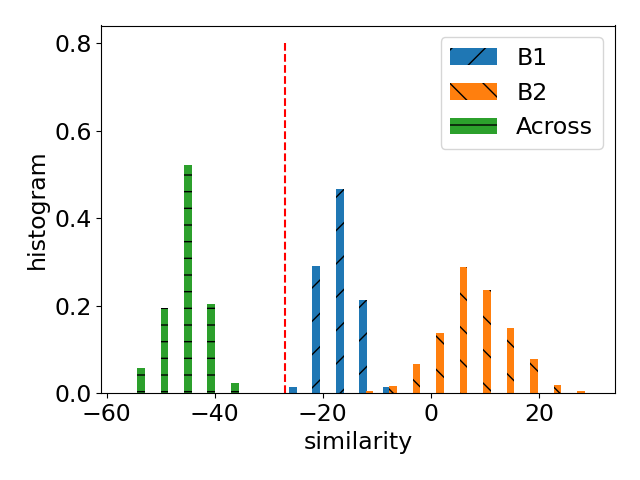}
        \caption{TW}
    \end{subfigure} \hspace*{0.02\textwidth}
    \begin{subfigure}[]{0.22\linewidth}
        \centering
        \includegraphics[width=\linewidth]{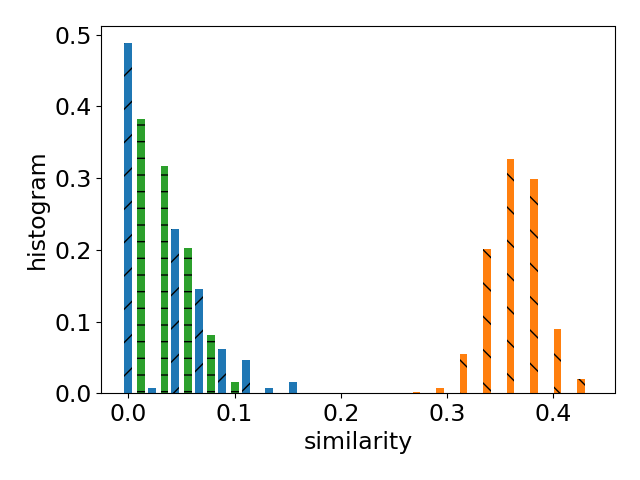}
        \caption{Tectonic}
    \end{subfigure}   \hspace*{0.02\textwidth}
    \begin{subfigure}[]{0.22\linewidth}
        \centering
        \includegraphics[width=\linewidth]{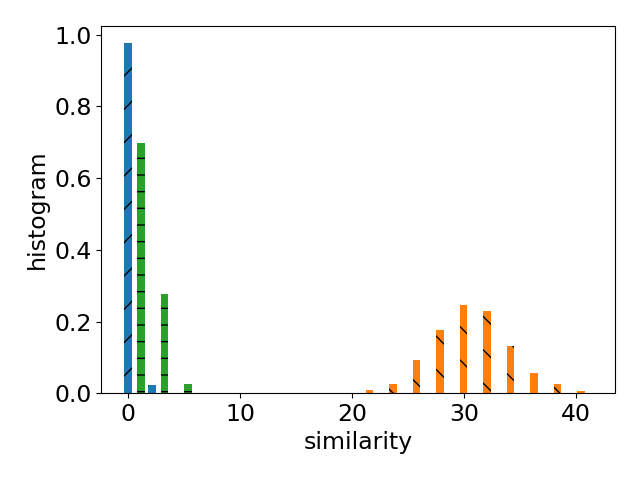}
        \caption{$K_3$}
    \end{subfigure}  \hspace*{0.02\textwidth}
    \begin{subfigure}[]{0.22\linewidth}
        \centering
        \includegraphics[width=\linewidth]{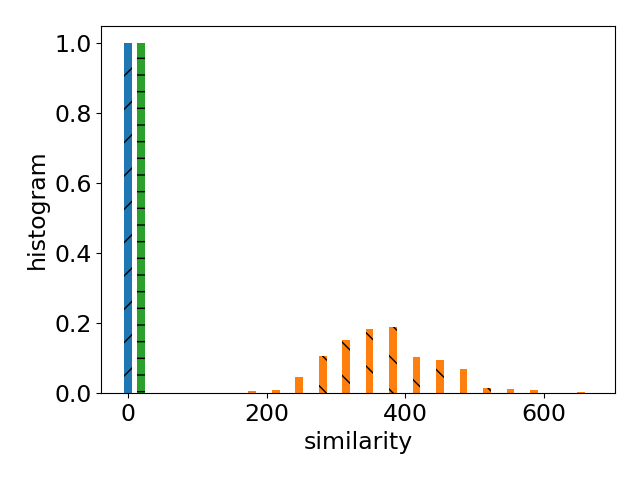}
        \caption{$K_4$}
    \end{subfigure}  \hspace*{0.02\textwidth}
    
    \begin{subfigure}[]{0.22\linewidth}
        \centering
        \includegraphics[width=\linewidth]{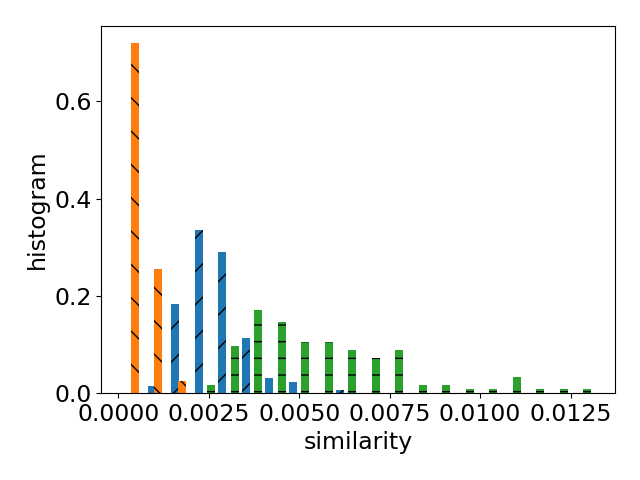}
        \caption{BC}
    \end{subfigure}  \hspace*{0.02\textwidth}
    \begin{subfigure}[]{0.22\linewidth}
        \centering
        \includegraphics[width=\linewidth]{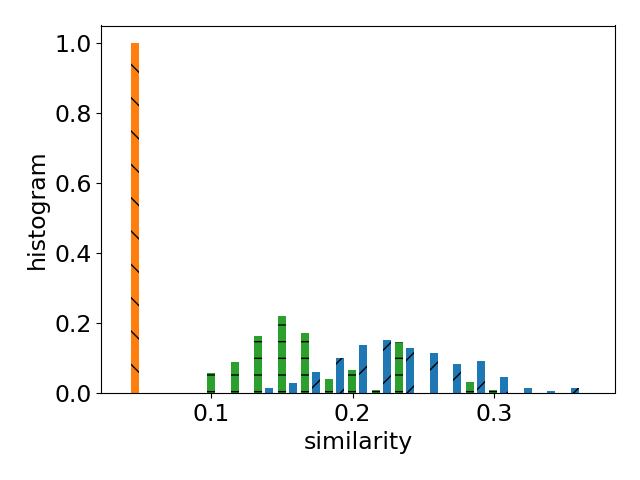}
        \caption{Eff. Res.}
    \end{subfigure}  \hspace*{0.02\textwidth}
    \begin{subfigure}[]{0.22\linewidth}
        \centering
        \includegraphics[width=\linewidth]{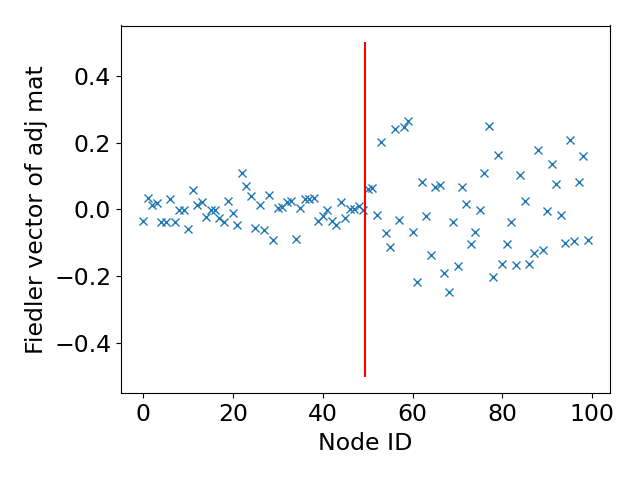}
        \caption{Spectral clustering}
    \end{subfigure}  \hspace*{0.02\textwidth}
    \begin{subfigure}[]{0.22\linewidth}
        \centering
        \includegraphics[width=\linewidth]{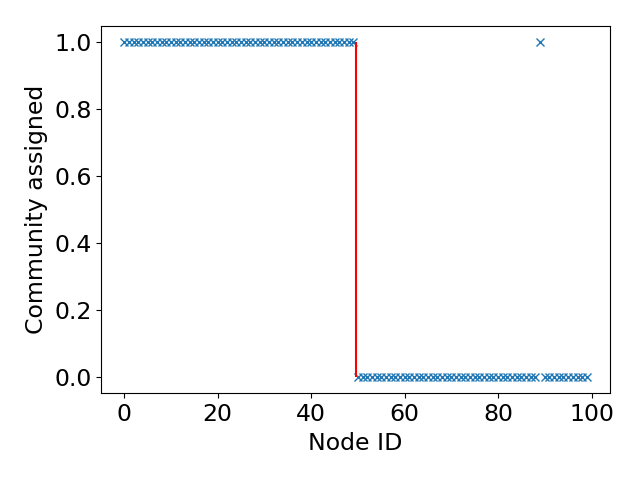}
        \caption{Louvain}
    \end{subfigure} \hspace*{0.02\textwidth}
    
    \caption{Comparison of community detection methods regarding output quality on an imbalanced SBM. Only \TW can fully recover the groundtruth communities. Legends are shared across subfigures.}
    \label{fig:syn_compare}
\end{figure*}

\begin{proof}
    
Recall that Tectonic removes  an edge $(u,v)$  if $\frac{t(u,v) }{d(u)+d(v)} < \delta $. Thus, we   define the tectonic score $Tec_{\delta}(u,v) = t(u,v) - \delta (d(u)+d(v))$. Table~\ref{tab:expected} provides a summary of the similarity score of each edge in expectation for the Tectonic and \TW scores, respectively, depending on the block participation of its endpoints $u,v$.  

\noindent {\it Tectonic:} Observe that Tectonic recovers the blocks in expectation if the similarity scores of the edges across the two blocks and the edges within blocks are well separated by some threshold. Assume there is such a threshold $\delta$ satisfies the condition.  This is equivalent to having $E[Tec_{\delta}(u,v)]>0$ for any $u,v\in B_i$, and $E[Tec_{\delta}(u,v)]<0$ for any $u\in B_1,v\in B_2$. Specifically, for any edge across two communities,

\begin{align*}
     E[Tec_{\delta}(u,v); u \in B_1, v \in B_2] &  = p_2q+2q^2-\delta(p_2+4q) <0  \rightarrow  \\ 
    \delta &> \frac{p_2q+2q^2}{p_2+4q} \overset{n \to \infty}{\longrightarrow} q
\end{align*}

\noindent For any edge induced by block $B_1$,

\begin{align*}
    E[Tec_{\delta}(u,v); u,v \in B_1] &  = 5q^2 - 6\delta  q
\end{align*}

\noindent By choosing $\delta>q$, 
\begin{align*}
    E[Tec_{\delta}(u,v); u,v \in B_1] & < -q^2<0
\end{align*}

\noindent Thus, Tectonic fails.

\underline{TW:} We will show that a range of thresholds exists to recover the communities in expectation. This is equivalent to having a positive gap in expectation between \TW of edges inside the block $B_1$ and across the two blocks. Specifically, the expected gap is $n(p_2-p_1) - 3n(p_1(q-p_1) + q(p_2-q))$. Given the assumption that $p_1=o(p_1)$, this gap value is dominated by $np_2$, making it positive. Therefore, TW can recover the communities in expectation.
 
\end{proof}

We visualize the analysis in Figure~\ref{fig:syn_compare}. Graphs are sampled from SBM with $n=50,p_1=0.1,p_2=0.8,q=0.05$. In Figure (a), we observe that for \TW, there is a threshold marked with a red dashed line that can clearly distinguish edges within and across the blocks. For \tec, the similarity scores of edges across the blocks are higher than the ones within the sparser block $B_1$, aligning with our analysis that Tectonic cannot recover the blocks in this case. A similar pattern is observed for $K_3$. Any higher-order clique \textit{sim} function shows no meaningful result as edges across the blocks and within the block $B_1$ have no participation. In addition, we test the ability of two other similarity functions that are not motif-based, i.e., betweenness centrality and effective resistance. We also evaluate spectral clustering, which has theoretical guarantee by Cheeger's inequality. Empirical results show only edge betweenness centrality can partially reveal the community structure by detecting edges across communities with higher similarity scores. Spectral clustering fails to return any insightful result due to the sparsity of $B_1$.  Louvain, one of the most frequently used community detection algorithms in practice, almost perfectly recovers the community structure (i.e., except for a couple of nodes).


\subsection{Parallel implementation} 


The motif-based framework consists of two components, i.e., computing the edge similarity scores and finding the connected components in the reweighted graph. For most of the scalable similarity functions highlighted in Table~\ref{tab:simfunctions}, the primary challenge in parallelization lies in the computation of triangle (or wedge) participation for edges. As a preprocessing step, we first sort the adjacency list by node id (wlog $V=[n]$), which takes $O(m\log m)$ work and $O(\log^{3/2}m)$ depth. Then, for each edge $(u,v)$, we directly compute its triangle participation as the size of neighborhood intersection $|N(u)\cap N(v)|$ and distribute the work for all edges in parallel. As mentioned in Section~\ref{sec:related}, computing the intersection of the neighborhoods with a sorted adjacency list takes $O(n)$ work and $O(\log n)$ depth. Note for computing the wedge participation, we need an extra step to compute the size of the neighborhood union, which takes the same work and depth as the intersection and thus does not affect the overall complexity. To compute the connected components, we apply the parallel randomized algorithm proposed by Shun et al.~\cite{shun2014simple} that takes linear work and $O(\log^3 n)$ depth with high probability. In all, our implementation of \TW requires $O(nm)$  work and $O(\log^3 n)$ depth. 

While the state-of-the-art algorithm for triangle counting takes $O(m^{3/2})$ work and $O(\log^{3/2}m)$ depth in the worst-case analysis, we argue such an algorithm requires extra memory in the case of the motif-based framework to count the number of triangles that participated for each edge. Assume triangle counts of all edges are stored in an array $C$ of $m$ entries. For each triangle composed of $e_1,e_2,e_3$, we have to update the array $C$ three times by adding one to each corresponding entry. To parallelize the work in counting and avoid writing collisions, we have to extend the array $C$ into two dimensions, i.e., $C\in \mathcal{Z} ^{|E|\times k}$ where $k$ is the number of workers. This can easily cause memory overflow when the graph size is large.

%% file: exp.tex
\subsection{Setup}



\noindent \textbf{Datasets.} We assess our framework and other competing approaches on various real-world graphs, both with and without known community structures. The characteristics of these graphs are detailed in Tables~\ref{tab:datasum}.  For the datasets with ground truth, we utilize six social and information networks from~\cite{yang2012define} that come with their respective overlapping groundtruth communities. In addition, we include Gavin~\cite{gavin2006proteome} and EXTE~\cite{Krogan2006GlobalLO1}, which are protein-protein interaction (PPI) networks of yeast cells and bacteria, respectively.  

\begin{table}
	\scriptsize
\centering
	\caption{Dataset statistics summary}
	\label{tab:datasum}
	\begin{tabular}{c|c|c|c|c}
		Name & \# nodes & \# edges & avg. degree &  \# communities \\
		\midrule
		Gavin & 1\,727 & 7\,534 & 8.72 & 247 \\
		EXTE & 3\,642 & 14\,300 & 7.84 & 307 \\
		Amazon (AM) & 334\,863 & 925\,872  & 5.52  & 40\,192 \\
		DBLP (DB) & 317\,080 & 1\,049\,866  & 6.62 & 13\,477 \\
		Youtube (YT) & 1\,134\,890 & 2\,987\,624  & 5.26 & 4\,498 \\
		LiveJournal (LJ) & 3\,997\,962 &	34\,681\,189  & 17.34 & 145\,133 \\
            Orkut (OR) & 3\,072\,441 & 117\,185\,083  & 76.28 & 903\,707 \\
		Friendster (FR) & 65\,608\,366 &	1\,806\,067\,135 & 55.06  & 650\,021 \\ \hline 
          Government~\cite{snapnets}	& 7\,057	& 89\,455 & 25.35 & - \\
        New-sites & 27\,917 & 206\,259 & 14.78 & -  \\
        Co-purchase & 403\,394 & 3\,387\,388 & 16.79 &- \\
        \begin{tabular}{@{}c@{}}Commoncrawl-  \\ host~\cite{commoncrawl}\end{tabular}  & 89\,247\,739  & 1\,940\,007\,864  & 43.47 & - \\
	\end{tabular}
\end{table}

\noindent \textbf{Evaluation.} There is a subtle issue in the evaluation. Some datasets come with overlapping groundtruth communities but all algorithms applied output a non-overlapping partition. To this end, we design the following framework that matches communities by maximizing Jaccard similarity. Specifically, given a set of groundtruth communities $\mathcal{C}=\{C_1,C_2 \dots, C_k \}$ and non-overlapping communities output by an algorithm $\mathcal{S}=\{S_1, S_2\dots, S_{k'} \}$, we find for each $S_i\in \mathcal{S}$ a groundtruth community $C^\star_k$ that maximizes the  Jaccard similarity, i.e., $C_k^\star = \arg\max_{C_j \in \mathcal{C}}\frac{|C_j \cap S_i|}{|C_j\cup S_i|}$. Precision and recall are then computed between $C_k^\star$ and $S_i$ as $p_i=\frac{|C_k^\star \cap S_i|}{|S_i|}$ and $r_i=\frac{|C_k^\star\cap S_i|}{|C_k^\star|}$, respectively. After matching all $S_i$'s with their groundtruth communities, the final evaluation metrics are calculated via weighted averaging by size, i.e., precision as $\frac{\sum_i p_i \cdot |S_i|}{\sum_i |S_i|}$ and recall as  $\frac{\sum_i r_i \cdot |S_i|}{\sum_i |S_i|}$. F1 score is computed as $\frac{\sum_i 2\cdot p_i\cdot r_i\cdot |S_i|/(p_i+r_i) }{\sum_i |S_i|}$. Remark that, unlike previous works~\cite{Tsourakakis2017scalable,shi2021scalable}, our evaluation framework does not guarantee that precision and recall will be monotone concerning the resolution parameter. In the Appendix, we provide a toy example that illustrates this claim. 

\noindent \textbf{Groundtruth communities for SNAP datasets.} For each SNAP dataset, we evaluate any result based on the whole set of groundtruth communities, except the ones with a size of less than 3. Note this is different from the previous studies that usually take the top 5000 communities as the ground truth. This is because the metrics used in~\cite{yang2012define} to find top communities are separability, density, cohesiveness, and clustering coefficient. Such measurements naturally constrain the community structures to follow their hypothesis and lose the value of using ground truth data. For example, Figure~\ref{fig:dblp_com_density_compare} shows more than one-third of the top 5000 communities selected in the three datasets are cliques or near-cliques, i.e., subgraphs with edge densities close to 1. In contrast, the real ground truth communities have edge densities ranging more widely. Therefore, the downstream evaluations favor methods that generate clique-like community structures. 

\begin{figure*}
    \centering
    \begin{subfigure}[]{0.3\linewidth}
        \centering
        \includegraphics[width=\linewidth]{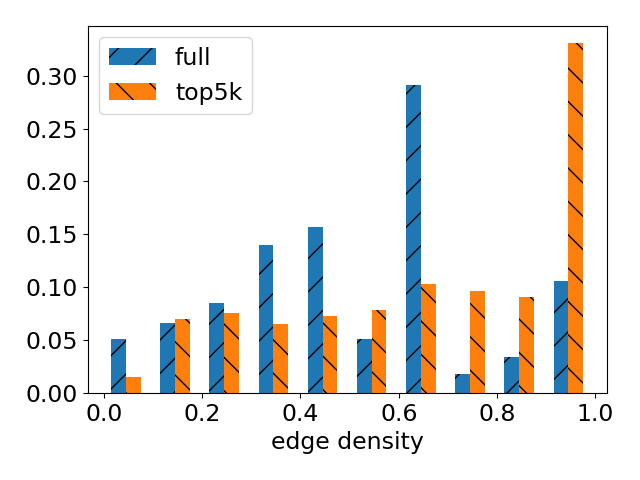}
        \caption{}
    \end{subfigure} \hfill
    \begin{subfigure}[]{0.3\linewidth}
        \centering
        \includegraphics[width=\linewidth]{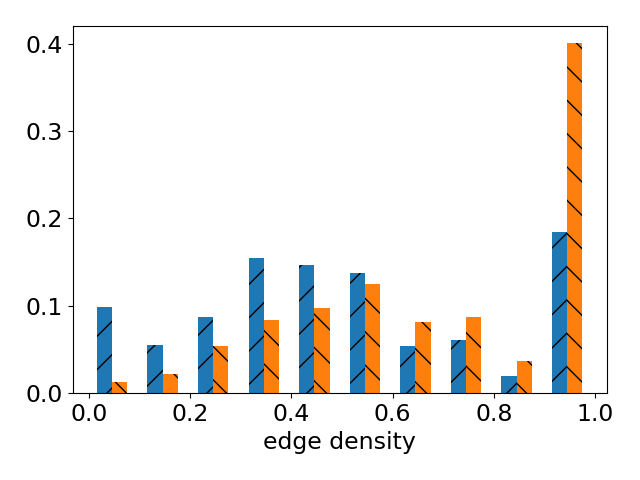}
        \caption{}
    \end{subfigure} \hfill
    \begin{subfigure}[]{0.3\linewidth}
        \centering
        \includegraphics[width=\linewidth]{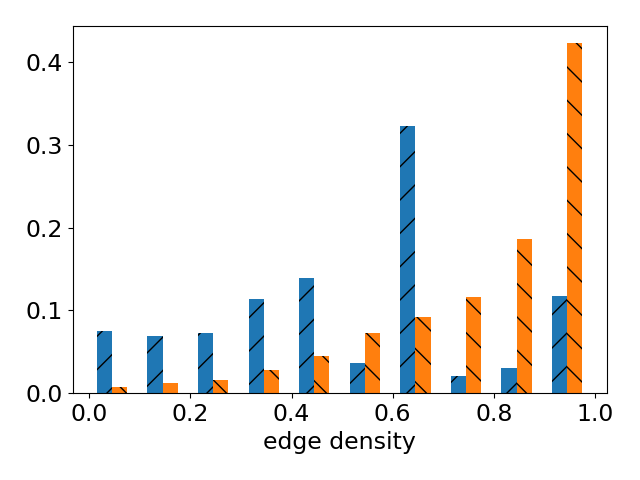}
        \caption{}
    \end{subfigure}
    \caption{Histograms of community densities of (a) Amazon, (b) DBLP, (c) LiveJournal.}
    \label{fig:dblp_com_density_compare}
\end{figure*}



\noindent \textbf{Baselines.} Apart from the motif-based methods with different similarity functions introduced in Table~\ref{tab:simfunctions}, we also evaluate a set of community detection methods as follows. Resolutions are tuned by parameters listed in Table~\ref{tab:params}. Other parameters are set by default.

\begin{itemize}
    \item {\it Parallel correlation clustering (PCC)}~\cite{shi2021scalable} . This method greedily optimizes the LambdaCC objective proposed by Veldt et al.~\cite{veldt2018correlation}. Authors optimize their parallel framework with various options, e.g., asynchronous setting, multi-level refinement, etc. 
    \item {\it Parallel modularity optimization (PMod)}~\cite{shi2021scalable}. The framework of Shi et al.~\cite{shi2021scalable} also supports optimizing in parallel a modularity objective proposed by Reichardt and Bornholdt~\cite{reichardt2006Statistical}, which is a generalization of the standard modularity metric used by Louvain~\cite{Girvan2002community}. 
    \item {\it Label propagation communities (LPC)}. We use the implementation provided by Networkx, which is a semi-synchronous label propagation method~\cite{cordasco2010community}.
    \item {\it Infomap}. The package provided by~\cite{mapequation2023software} is written in C++ and wrapped with Python APIs.
    \item {\it Spectral clustering\footnote{\url{https://scikit-learn.org}.} (Spec)}~\cite{ng2001spectral}, and higher-order spectral clustering using triangles (H-Spec)~\cite{benson2016higher}.
\end{itemize}

\begin{table}[th]
    \centering
        \caption{  \label{tab:params}Parameters and their respective ranges explored for each method to control resolution when evaluating the precision-recall tradeoff.}
    \begin{tabular}{c|c|c} 
      Method   &  Parameter & Range (start, end, step size) \\ \midrule
       TW  & $\delta$ & (-30,0,2) \\ 
       Tectonic  & $\delta$ & (0,0.3,0.02) \\ 
       $K_3$  & $\delta$ & (0,15,1) \\ 
       PCC  & $\lambda$ & (0.01,0.96,0.05) \\ 
       PMod  & $\gamma$ & $2m\times (0.1,1,0.05)$ \\ 
       Spectral clustering & \# of clusters & (100,1300,300)
    \end{tabular}

\end{table} 

\begin{table*}
\centering
\footnotesize
\caption{Time and memory usage of TW, PCC and PMod on rMat random graphs ranged by density and edge size. We highlight the cell corresponding to the method with the best performance under each setting.}
\begin{tabular}{l|l|l|lll|lll}
                              &      &      & \multicolumn{3}{c}{Time(second)}                                                                        & \multicolumn{3}{c}{Memory(GB)}                      \\ \midrule
                              & n    & m    & TW                             & PCC                            & PMod                          & TW                            & PCC    & PMod   \\ \midrule
                              & 2M   & 10M  & \cellcolor{green!30}1.77   & 5.51                           & 6.63                          & \cellcolor{green!30}0.6   & 2.37   & 2.47   \\
                              & 20M  & 100M & \cellcolor{green!30}20.05  & 53.39                          & 111.72                        & \cellcolor{green!30}6.19  & 21.63  & 26.13  \\
\multirow{-3}{*}{very sparse ($m=5n$)} & 200M & 1B   & \cellcolor{green!30}164.19 & 700.67                         & 1759.88                       & \cellcolor{green!30}64.07 & 219.08 & 284.36 \\ \midrule
                              & 200K & 10M  & 4.96                           & 4.24                           & \cellcolor{green!30}2.85  & \cellcolor{green!30}0.52  & 1.24   & 1.67   \\
                              & 2M   & 100M & 71.06                          & 64.70                          & \cellcolor{green!30}42.94 & \cellcolor{green!30}5.39  & 12.75  & 15.87  \\
\multirow{-3}{*}{sparse ($m=50n$)}      & 20M  & 1B   & \cellcolor{green!30}566.64 & 1065.29                        & 978.40                        & \cellcolor{green!30}55.99 & 205.49 & 197.33 \\ \midrule
                              & 46K  & 10M  & 13.91                          & \cellcolor{green!30}1.18   & 2.25                          & \cellcolor{green!30}0.46  & 1.02   & 1.54   \\
                              & 215K & 100M & 378.11                         & \cellcolor{green!30}13.68  & 22.16                         & \cellcolor{green!30}4.73  & 9.07   & 12.56  \\
\multirow{-3}{*}{dense ($m=n^{3/2}$)}       & 1M   & 1B   & 9492.5                         & \cellcolor{green!30}131.15 & 252.99                        & \cellcolor{green!30}48.73 & 93.96  & 127.16
\end{tabular}
\label{tab:syn}
\end{table*}

\noindent \textbf{Machine Specs.} We run most experiments on cloud instances with two 16-core 2.8 GHz Intel Gold 6242 processors and 384 GB of main memory. Our code is written in C++. We compile our programs with g++ and the -O3 flag and use an efficient work-stealing scheduler~\cite{Blelloch2020parallel}. Max memory usage is reported via \textit{getrusage} function from the Standard C library of Linux. We also have a sequential implementation of most methods written in Python that can be run on graphs except for graphs with more than 100M edges. Our Git repo can be found 
\href{https://github.com/tsourakakis-lab/parallel-motif-communities}{here}. 


\subsection{Scalability.}

We evaluate the scalability of our parallel framework based on the time and memory usage on both synthetic and real-world graphs. 

\noindent \textbf{\rmat random graphs.} We generate synthetic \rmat graphs~\cite{Chakrabarti2004RMATAR} with varying sizes and densities. Specifically, the number  of edges ranges in $\{10^7,10^8,10^9\}$ and the densities vary from  very sparse ($m=5n$), sparse ($m=50n$) to dense ($m=n^{3/2}$). The seed parameters of the $2\times 2$ matrix $A=[a_{ij}]$ are fixed as $a_{11}=0.45$, $a_{12}=a_{21}=0.15$, and $a_{22}=0.25$.

We report in Table~\ref{tab:syn} the running time and memory usage of our framework and two highly scalable parallel community detection methods PCC and PMod. The number of threads are set as 16 for all methods, and the method with the best performance is highlighted under each setup. The motif-based framework shows up to $10\times$ speedups compared to PMod on the very sparse regime, being comparable to both baselines in the sparse regime, and under-performs in the dense regime. This is expected with the \TW similarity function, as our framework computes for each edge the number of triangles and wedges it participates in.  In this type of regimes practitioners reside on sparsification methods such as Doulion~\cite{tsourakakis2009doulion}
.  On the other hand, our framework is consistently the most memory-efficient, using around $3\times$ less memory usage.

\noindent \textbf{Real-world graphs.} We report the running time of our parallel framework and two baselines on real-world networks in Figure~\ref{fig:runtime}. The number of workers are ranged in $\{1,2,4,8,16,32\}$. The running time of our framework is shorter compared to both baselines on the sparse networks, i.e., Amazon, DBLP and Youtube. In addition, both components of our framework are highly parallelizable, reflected in the almost linear relation between the number of workers and the running time, as indicated by the dashed line that shows ideal linear speedups as a function of the workers. Such high parallelization supports the efficiency of our framework even on large and dense graphs. For networks like Commoncrawl and Friendster, the execution times are highly reduced by increasing the number of works, gradually becoming comparable to the competitors. 

Furthermore, we display the relation between memory usage and the graph size in Figure~\ref{fig:runtime} (h). While all methods show linear memory usage increment with respect to the number of edges, the baselines incur $1.6\times$ to $8\times$  memory overheads across datasets.

\begin{figure*}
    \centering
    \begin{subfigure}[]{0.22\linewidth}
        \centering
        \includegraphics[width=\linewidth]{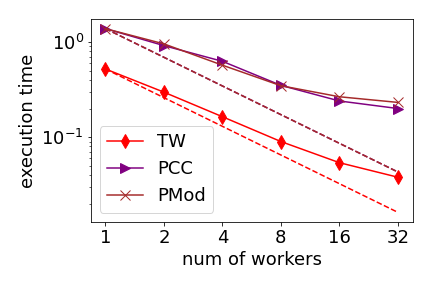}
        \caption{Amazon}
    \end{subfigure} \hfill
    \begin{subfigure}[]{0.22\linewidth}
        \centering
        \includegraphics[width=\linewidth]{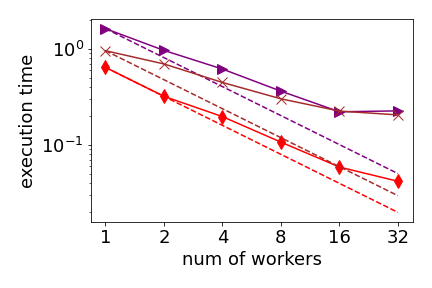}
        \caption{DBLP}
    \end{subfigure} \hfill
    \begin{subfigure}[]{0.22\linewidth}
        \centering
        \includegraphics[width=\linewidth]{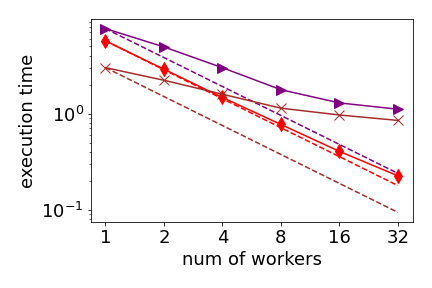}
        \caption{Youtube}
    \end{subfigure} \hfill
    \begin{subfigure}[]{0.22\linewidth}
        \centering
        \includegraphics[width=\linewidth]{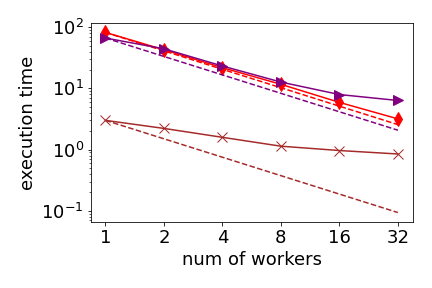}
        \caption{LiveJournal}
    \end{subfigure} \hfill
    \begin{subfigure}[]{0.22\linewidth}
        \centering
        \includegraphics[width=\linewidth]{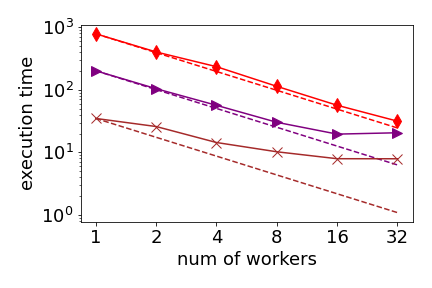}
        \caption{Orkut}
    \end{subfigure} \hfill
    \begin{subfigure}[]{0.22\linewidth}
        \centering
        \includegraphics[width=\linewidth]{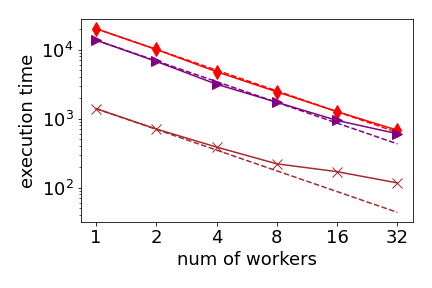}
        \caption{Friendster}
    \end{subfigure}\hfill
    \begin{subfigure}[]{0.22\linewidth}
        \centering
        \includegraphics[width=\linewidth]{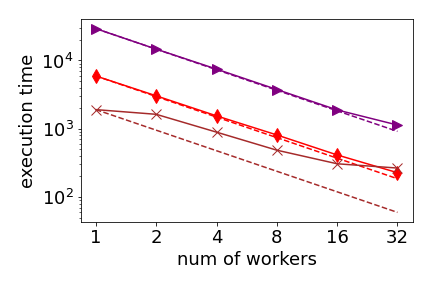}
        \caption{Commoncrawl}
    \end{subfigure} \hfill
    \begin{subfigure}[]{0.22\linewidth}
        \centering
        \includegraphics[width=\linewidth]{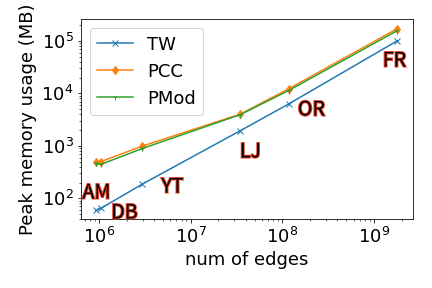}
        \caption{Memory usage}
    \end{subfigure} 
    \caption{(a)-(g) Running time on real-world graphs. Dashed lines represent the sequential running time divided by the number of workers, i.e., expected running time without parallel cost. (h) Memory usage on SNAP graphs.}
    \label{fig:runtime}
\end{figure*}


\begin{figure*}[h]
    \centering
    \begin{subfigure}[]{0.24\linewidth}
        \centering
        \includegraphics[width=\linewidth]{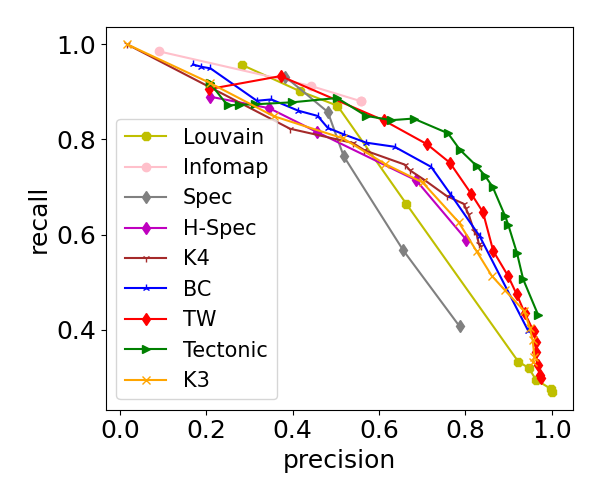}
        \caption{Gavin}
    \end{subfigure} 
    \begin{subfigure}[]{0.24\linewidth}
        \centering
        \includegraphics[width=\linewidth]{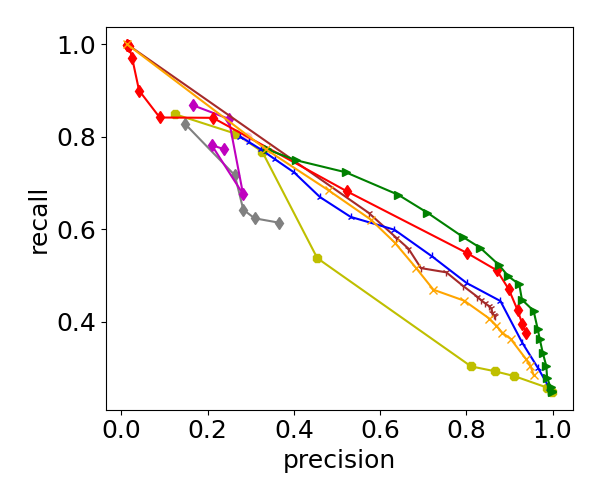}
        \caption{EXTE}
    \end{subfigure}
    \begin{subfigure}[]{0.24\linewidth}
        \centering
        \includegraphics[width=\linewidth]{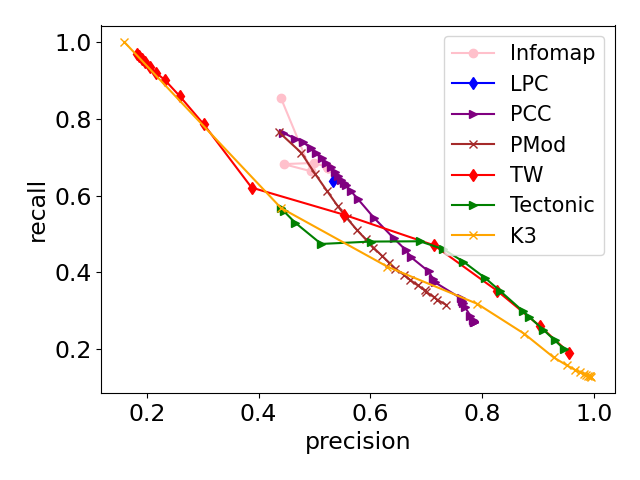}
        \caption{Amazon}
    \end{subfigure} \hfill
    \begin{subfigure}[]{0.24\linewidth}
        \centering
        \includegraphics[width=\linewidth]{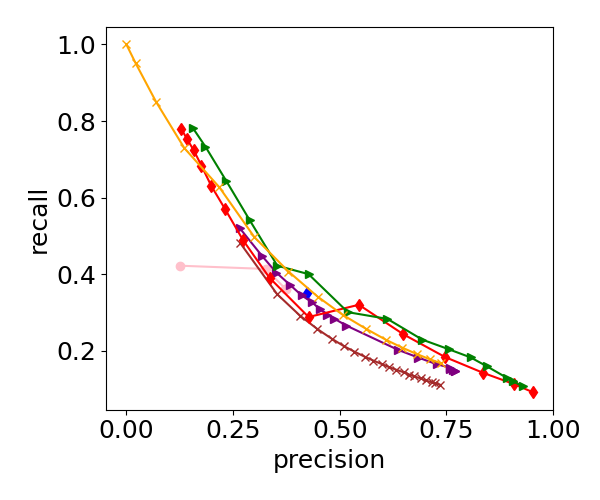}
        \caption{DBLP}
    \end{subfigure} \hfill
    \begin{subfigure}[]{0.24\linewidth}
        \centering
        \includegraphics[width=\linewidth]{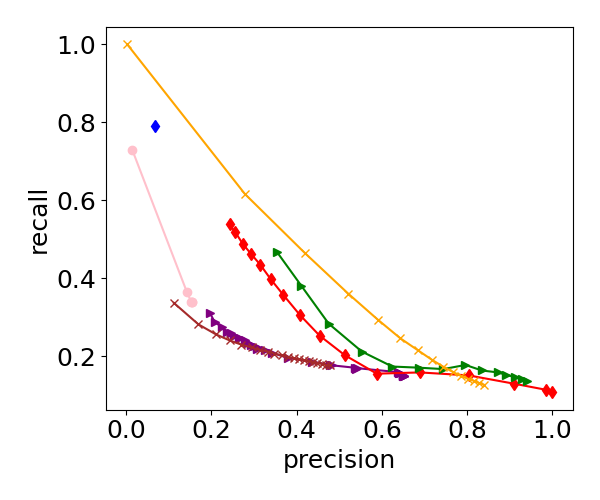}
        \caption{Youtube}
    \end{subfigure}
    \begin{subfigure}[]{0.24\linewidth}
        \centering
        \includegraphics[width=\linewidth]{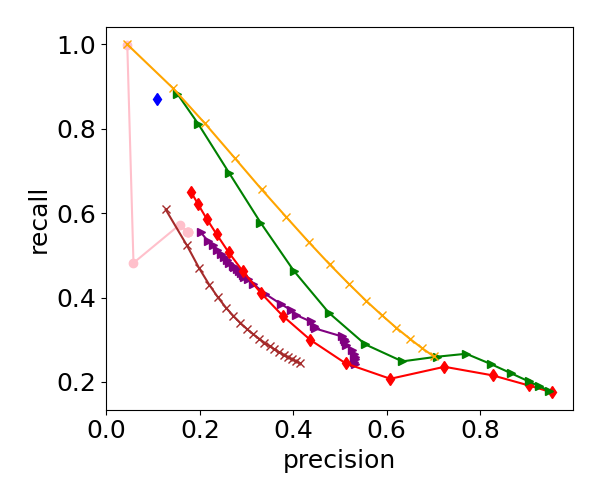}
        \caption{LiveJournal}
    \end{subfigure} \hfill
    \begin{subfigure}[]{0.24\linewidth}
        \centering
        \includegraphics[width=\linewidth]{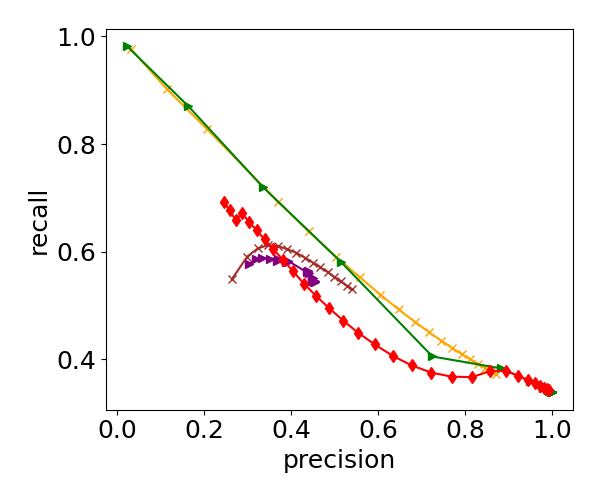}
        \caption{Orkut}
    \end{subfigure} \hfill
    \begin{subfigure}[]{0.24\linewidth}
        \centering
        \includegraphics[width=\linewidth]{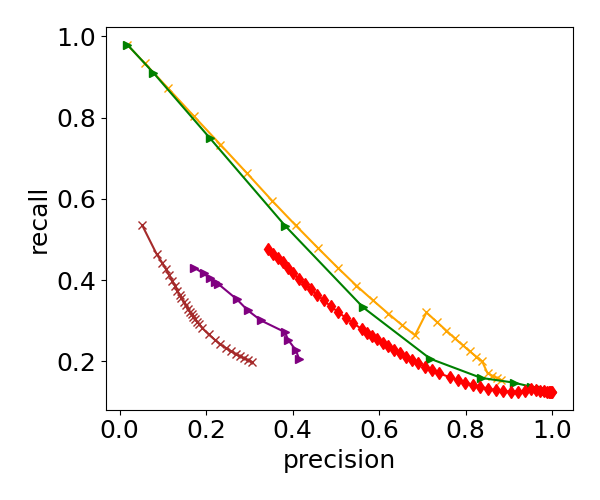}
        \caption{Friendster}
    \end{subfigure}
    \caption{Precision-recall tradeoff on graphs with groundtruth communities.}
    \label{fig:precision}
\end{figure*}

\subsection{Quality Evaluation via Groundtruth Communities}

\noindent \textbf{PPI networks.} Two PPI networks we test are small in size and thus can be used to evaluate methods that are not scalable, e.g., (higher-order) spectral clustering and our motif-based framework with betweenness centrality as the similarity function. We demonstrate the precision-recall tradeoff of methods in Figure~\ref{fig:precision}(a), (b). We observe \tec followed by \TW shows the top performance among all the baselines. While Infomap achieves one of the best performances the recall-precision tradeoff is less flexible than the other methods. This shows the high sensitivity of Infomap concerning its resolution parameter.

\noindent \textbf{SNAP networks.}
Figure~\ref{fig:precision} (c)-(h) shows the precision-recall tradeoff of various community detection methods on SNAP datasets with groundtruth communities using all groundtruth communities. No method dominates all other methods across all networks. For example, in Amazon, \TW and PCC achieve the top performance, followed by Infomap and label propagation. On the other hand, in Youtube and LiveJournal, our motif-based framework with triangle participation as the similarity function provides better accuracy than any other method. We observe that PCC and PMod exhibit a major performance downgrade compared to the original paper by Shi et al.~\cite{shi2021scalable}. Notice in conjunction with Figure~\ref{fig:dblp_com_density_compare} that many groundtruth communities are sparse. Thus, PCC and PMod, which typically find communities with high internal edge densities, tend to underperform.


In addition, we observe that the performance of both Tectonic and \TW drop when tuning the resolution from high to low. Specifically, both methods achieve high accuracies when the resolution is high (i.e., in the regime of low recall and high precision) but lose the advantage when the threshold $\delta$ (see Algorithm~\ref{alg:framework}) is reduced up to some extent. Such a phenomenon can be observed in all the SNAP datasets and is especially obvious in LiveJournal, where both precision and recall metrics drop significantly when the precision is in the range of $0.6-0.7$. We look into this phenomenon in greater detail in the next subsection.

\subsection{Threshold Selection}

\begin{figure*}[]
    \centering
    \begin{tabular}{p{0.02\linewidth}cccc}
       & DBLP & Government & New-sites & Co-purchase \\
        \begin{tabular}{l}
    \TW
  \end{tabular} & \begin{tabular}{c}\includegraphics[width=0.2\linewidth]{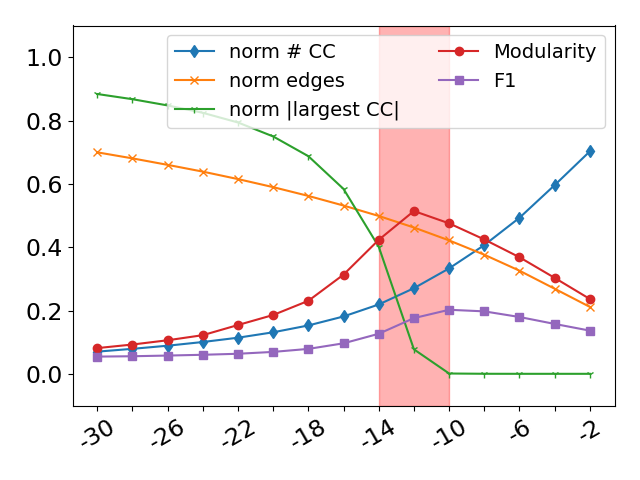} \end{tabular}  &  \begin{tabular}{c}\includegraphics[width=0.2\linewidth]{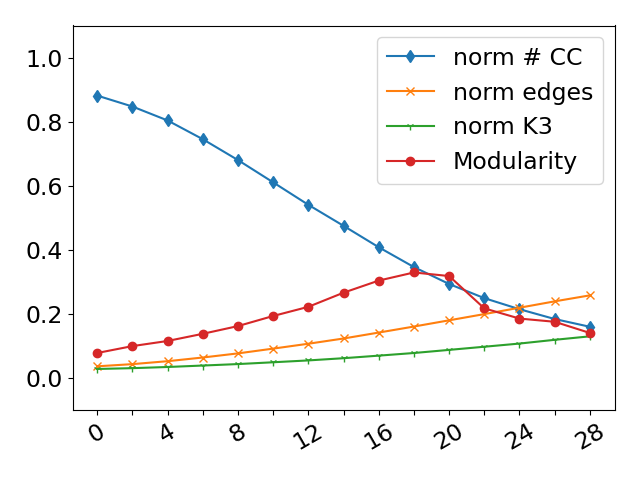}  \end{tabular} & \begin{tabular}{c} \includegraphics[width=0.2\linewidth]{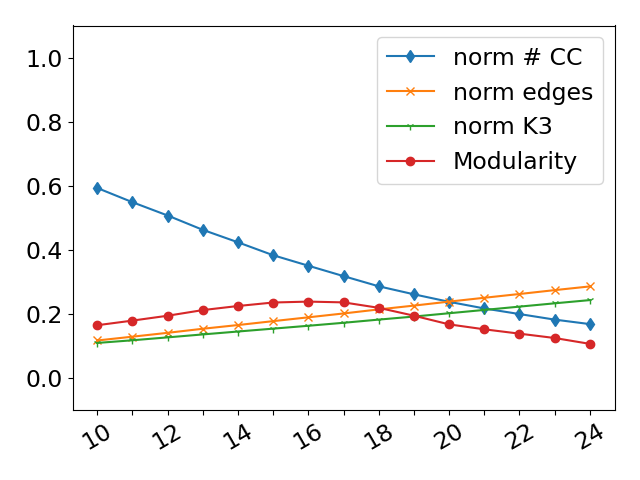}  \end{tabular} & \begin{tabular}{c} \includegraphics[width=0.2\linewidth]{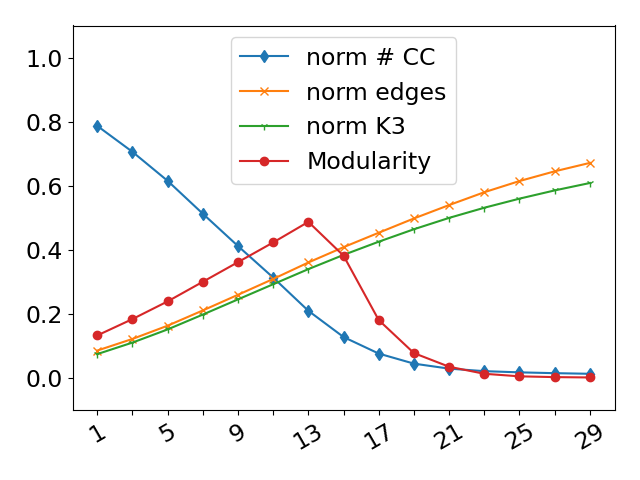}  \end{tabular} \\
        \begin{tabular}{l}
    Tec.
  \end{tabular} &  \begin{tabular}{c}\includegraphics[width=0.2\linewidth]{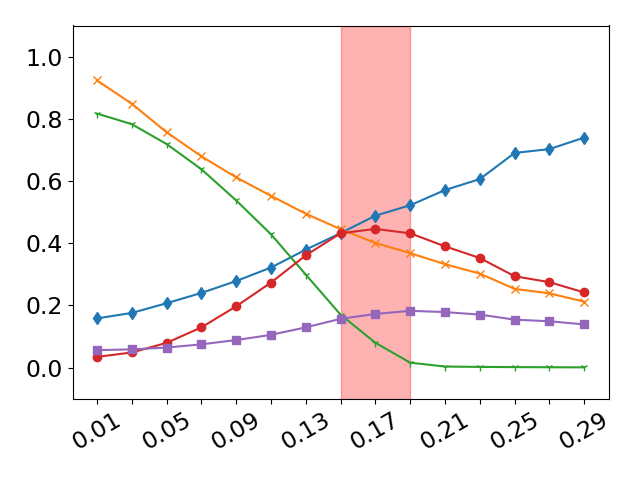} \end{tabular} &  \begin{tabular}{c}\includegraphics[width=0.2\linewidth]{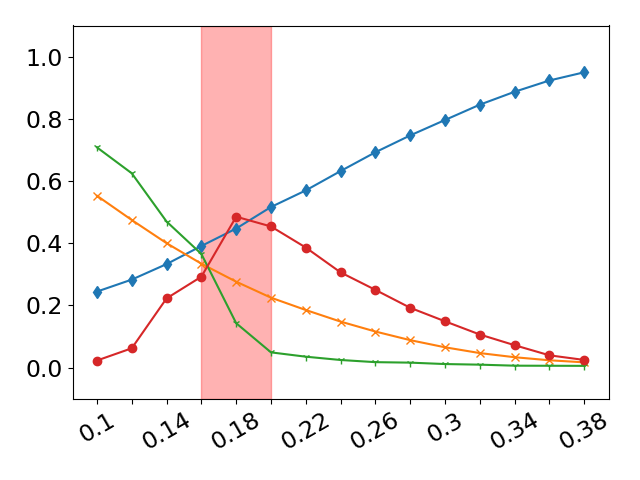} \end{tabular}&  \begin{tabular}{c}\includegraphics[width=0.2\linewidth]{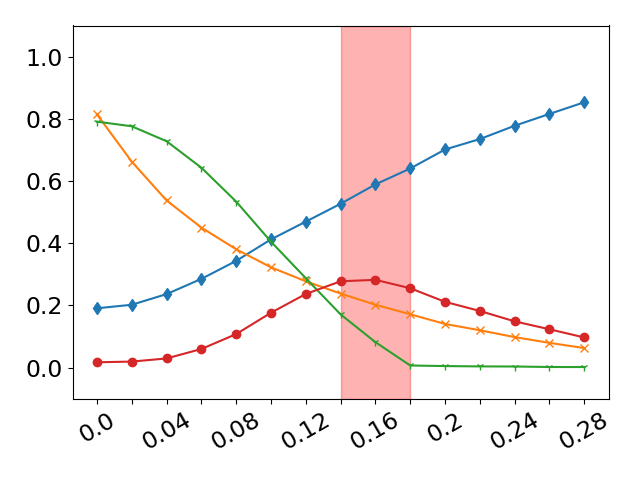}\end{tabular} &  \begin{tabular}{c}\includegraphics[width=0.2\linewidth]{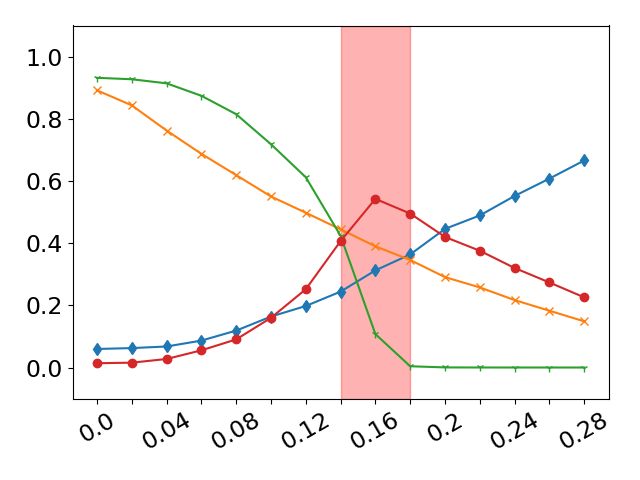}\end{tabular} \\
        \begin{tabular}{l}
    $K_3$
  \end{tabular} &  \begin{tabular}{c}\includegraphics[width=0.2\linewidth]{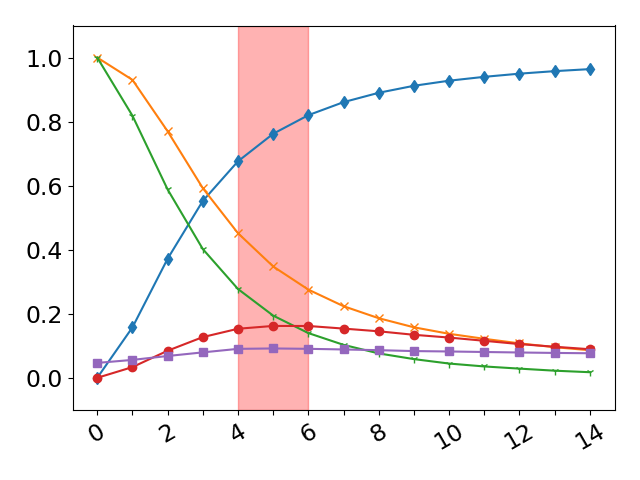}\end{tabular} &  \begin{tabular}{c}\includegraphics[width=0.2\linewidth]{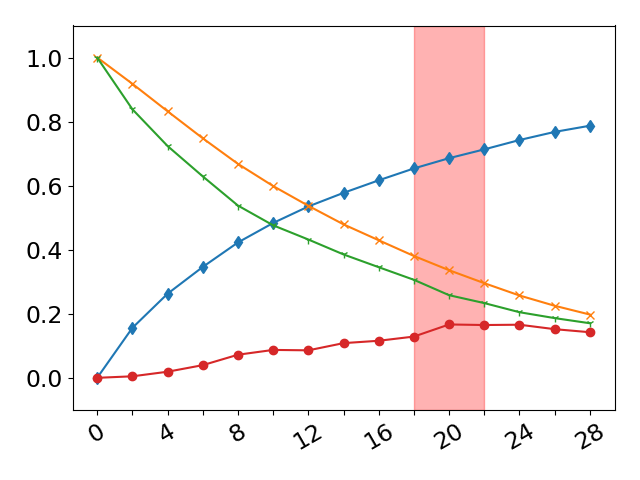} \end{tabular} &  \begin{tabular}{c}\includegraphics[width=0.2\linewidth]{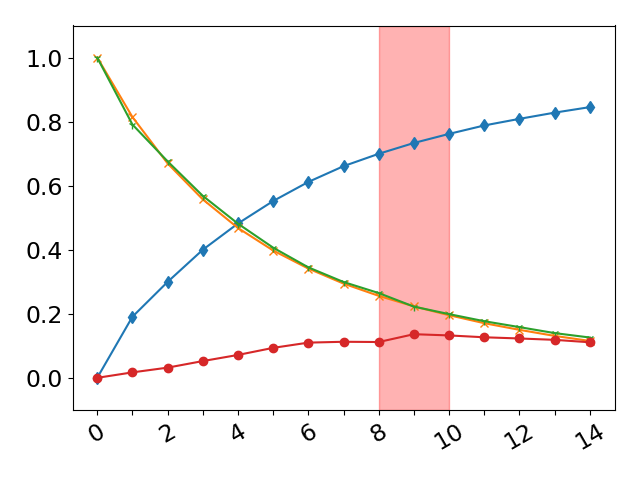} \end{tabular} &  \begin{tabular}{c}\includegraphics[width=0.2\linewidth]{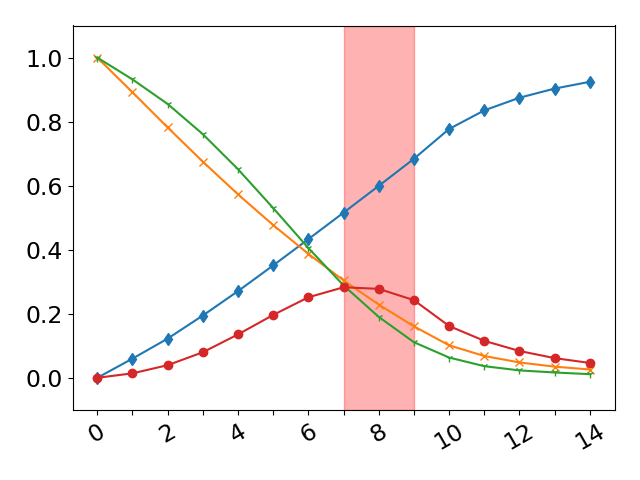} \end{tabular}\\
    \end{tabular}
    \caption{Threshold selection of TW, Tectonic (Tec.), and $K_3$. Figures present the threshold value $\delta$ (x-axis) versus normalized output community statistics (y-axis), with the red annotated threshold $\delta$ values selected by our rules of thumb.}
    \label{tab:selection}
\end{figure*}

We observe that for all the motif-based methods, the performance dropped significantly after reducing the resolution. To better understand the reason behind this, we show in Figure~\ref{tab:selection} various graph metrics versus the threshold value for \TW, Tectonic, and $K_3$ on graphs with and without groundtruth. To be able to compare across networks with different numbers of nodes and edges, we report modularity together with normalized versions of the graph statistics. In particular, we report the number of connected components normalized by $n$ (norm \# CC), the number of edges left after the edge removal normalized by $m$ (norm edges), and the size of the largest connected component normalized by $n$ (norm |largest CC|). On DBLP where the groundtruth communities are available, we also present the F1 score.

We observe consistently non-trivial jumps in the size of the largest connected component (CC) for \tec and \TW when reducing the threshold. In addition, these jumps match perfectly with the performance collapse observed in the previous subsection, which provides a reason for such a phenomenon.  Initially, when the resolution is high, i.e., the threshold is set high for a motif-based method, the input graph is disconnected into small subgraphs with decent precisions. However, a large network has different properties and structures across subgraphs. Specifically, some components in a network can possess higher density or connectivity than others, e.g., in a global road network, the component of Frankfurt can be much denser than the one in Sicily. When the resolution is reduced to some extent, some dense components cannot get disconnected and thus will be returned as a whole, while the sparse components are still split into accurate communities. Therefore, using a universal threshold for all the edges can cause an accuracy issue, especially when clustering with low resolution on a large network. 

On the other hand, the motif-based community detection methods provide impressive performance if we can avoid the relatively low-resolution regime. As we run the framework with {\it decreasing} thresholds, {\it a sudden increment in the size of the largest community} output by a motif-based method can serve as a good indicator. This can be used as a rule of thumb for picking threshold values for our method, which is an obstacle to using most community detection methods in practice.
In cases where we do not observe any unusual shift in the size of the largest community, e.g., thresholding based on $K_3$, we recommend {\it picking a threshold that maximizes the modularity}. As presented in Table~\ref{tab:selection}, both rules of thumb return very similar thresholds for Tectonic and \TW on various datasets.




We argue that while these rules may not lead to "optimal" community structures, where the optimality of a solution to the community detection problem is never easy to define and varies by application, our rules usually yield a decent output for further downstream tasks.

%% file: conclusion.tex
The framework of motif-aware community detection offers a blend of simplicity and practical efficacy. It operates by reweighting each edge based on a similarity measure between its endpoints and then discarding edges using a threshold. We introduce a novel edge similarity measure, \TW, which is based on counting triangles and wedges, both of which are extensively studied in various computational models, including the PRAM model. Our similarity measure is intuitively appealing, and we provide both theoretical and empirical support for its practical application in community detection. Specifically, we demonstrate that while \tec has shown promising empirical results, it can fail to identify well-defined communities, whereas \TW excels in this regard. Furthermore, we present an open-source parallel framework designed for scalable motif-based community detection methods. We assess the performance of our framework on both synthetic and real-world graphs, with and without known ground truth communities. Our experimental results indicate that our framework either outperforms or is on par with state-of-the-art baselines, such as Louvain and LambdaCC. Additionally, we suggest practical guidelines for selecting an effective threshold for real-world applications.